\documentclass[a4paper,UKenglish,cleveref,autoref,thm-restate,numberwithinsect]{lipics-v2021}
\pdfoutput=1

\hideLIPIcs
\nolinenumbers

\usepackage{booktabs}
\usepackage{graphicx}
\graphicspath{{figures/}}

\newcommand{\MainBound}{0.8344}                    
\newcommand{\ClassicalCeiling}{0.833597897}        
\newcommand{\CeilingDirs}{1{,}000{,}000}           
\newcommand{\HeadroomToCeiling}{\ensuremath{3.8\times10^{-4}}} 

\newcommand{\PalLB}{\ensuremath{\pi/8+\sqrt3/4}}   
\newcommand{\ElekesLB}{0.8271}                     
\newcommand{\BSLB}{0.832}                          
\newcommand{\XieLB}{0.833}                         
\newcommand{\GibbsUpper}{0.8440935944}             
\newcommand{\GibbsHeuristic}{0.83699098}           
\newcommand{\GibbsOursOnHis}{0.837296}             

\newcommand{\Nodes}{486{,}799{,}600}               
\newcommand{\Leaves}{245{,}496{,}952}              
\newcommand{\HardBoxes}{450{,}922{,}384}           
\newcommand{\HardMinutes}{364.9}                   
\newcommand{\EmitMinutes}{591.1}                   
\newcommand{\TreeMB}{60.85}                        
\newcommand{\BlockIndexMB}{16.78}                  

\newcommand{\CertMB}{81.3}                         
\newcommand{\BitsPerNode}{1.34}                    
\newcommand{\WorstSlack}{\ensuremath{1.2409\times10^{-5}}}   
\newcommand{\CrossCheckLeaves}{24{,}704}           
\newcommand{\CrossCheckWorst}{\ensuremath{6.05\times10^{-13}}} 

\newcommand{\Unit}{\ensuremath{2^{-53}}}           
\newcommand{\Tau}{\ensuremath{10^{-12}}}           
\newcommand{\EpsMem}{\ensuremath{1.0016\times10^{-12}}}  
\newcommand{\PerimBound}{8.026}                    
\newcommand{\ErrMembership}{\ensuremath{8.04\times10^{-12}}}  
\newcommand{\VertexBound}{3072}                    
\newcommand{\ShoelaceS}{5012}                      
\newcommand{\ErrShoelace}{\ensuremath{1.71\times10^{-9}}}     
\newcommand{\Eta}{\ensuremath{1.72\times10^{-9}}}             
\newcommand{\SlackOverEta}{\ensuremath{7.2\times10^{3}}}      
\newcommand{\HighPrecDiff}{\ensuremath{6.66\times10^{-15}}}   
\newcommand{\HighPrecDigits}{60}                   

\newcommand{\TmaxThree}{0.192366}                  
\newcommand{\TmaxFive}{0.195891}                   
\newcommand{\DomainShrink}{169}                    
\newcommand{\TmaxNaive}{0.70}                      
\newcommand{\CircThree}{0.577350269190}            
\newcommand{\CircFive}{0.525731112119}             

\newcommand{\AuditChecks}{29}                      
\newcommand{\KernelGates}{28}                      
\newcommand{\ControlsPassed}{22}                   
\newcommand{\PalAgreement}{\ensuremath{1.1\times10^{-16}}}    
\newcommand{\BSFigAgreement}{\ensuremath{3.2\times10^{-6}}}   

\newcommand{\BoxesCheapTarget}{25{,}425{,}568}     
\newcommand{\BoxesSharpDomain}{27{,}350{,}260}     

\newcommand{\ClassicalTthree}{(0.004667984,-0.008121781)}   
\newcommand{\ClassicalRho}{1.464915745}                     
\newcommand{\ClassicalTfive}{(0.012090029,-0.020585645)}    
\newcommand{\ReuleauxWitness}{(-0.012158832,\,-0.000000037,\,1.256661190,\,0.022996607,\,-0.000003897)}   
\newcommand{\ClassicalCeilingSafe}{0.8336}         
\newcommand{\FamilyBCeilingSafe}{0.83479}          
\newcommand{\ClassicalOuterA}{0.833638744}         
\newcommand{\ClassicalOuterB}{0.833601460}         
\newcommand{\ClassicalOuterC}{0.833597897}         
\newcommand{\ReuleauxOuterC}{0.834781191}          

\newcommand{\fAtTmax}{0.836548794}                 
\newcommand{\InradThree}{0.422649730810}           
\newcommand{\InradFive}{0.474268887881}            

\newcommand{\Inflate}{\ensuremath{10^{-11}}}       
\newcommand{\WitnessM}{1024}                       
\newcommand{\WitnessK}{1024}                       
\newcommand{\SeedDepth}{22}                        
\newcommand{\NSeed}{4{,}194{,}304}                 
\newcommand{\Hmin}{\ensuremath{10^{-5}}}           

\newcommand{\FiveDExpGlobal}{2.66}                 
\newcommand{\FiveDExpTight}{1.94}                  
\newcommand{\EightDExpGlobal}{5.50}                
\newcommand{\HalfDimFive}{2.5}                     
\newcommand{\HalfDimEight}{4.0}                    
\newcommand{\EightDOptimum}{0.836494901}           
\newcommand{\EightDBoxes}{\ensuremath{8.6\times10^{11}}}  
\newcommand{\EightDDays}{880}                      
\newcommand{\NextTargetBoxes}{\ensuremath{1.9\times10^{9}}}   
\newcommand{\NextTargetHours}{27}                  

\newcommand{\NextTarget}{0.8346}                  
\newcommand{\CheapTarget}{0.833}                  

\newcommand{\ThroughputFive}{19{,}648}             
\newcommand{\ThroughputEight}{11{,}351}            
\newcommand{\FiveDOptimum}{0.834780947}            
\newcommand{\RepoCommit}{da956b117a}            

\newcommand{\Rot}{\mathcal{R}}

\title{Curves of constant width and a lower bound for Lebesgue's universal
covering problem}

\titlerunning{Curves of constant width and Lebesgue's covering problem}

\author{Ujjwal Mishra}{Indian Institute of Information Technology Una, India}
       {ujjwalmishra238@gmail.com}{}{}

\authorrunning{U. Mishra}
\Copyright{Ujjwal Mishra}

\supplement{Software (Source Code):\\
\url{https://github.com/Ujjwal238/universal-cover-problem}}

\acknowledgements{The author thanks Philip Gibbs, whose 2014 paper proposed
curves of constant width as the test sets for this problem and explored them
numerically. The family used here is his, and the present work supplies a proof
for a choice he had already identified as the right one. The normalised placement
space and the strategy of bounding the hull area by subdivision are those of Peter
Brass and Mehrbod Sharifi, and the computation reported here sits inside the
framework they set out.}

\begin{document}
\maketitle

\begin{abstract}
  A universal cover is a convex set in the plane that contains a congruent copy of
  every planar set of diameter one. Lebesgue asked in 1914 for one of least area,
  and the value is not known. We prove that every convex universal cover has area
  at least 0.8344, improving on 0.832, published in 2005, and 0.833, in a 2026
  preprint, both of which come from a disc together with an equilateral triangle
  and a regular pentagon. Our test sets are instead curves of constant width: the
  disc, the Reuleaux triangle and the Reuleaux pentagon. Each contains the regular
  polygon it is built on, so the family is strictly larger at the same number of
  bodies and the same number of placement parameters, and we show that the
  classical configuration admits an arrangement whose hull has area below 0.8336,
  so no bound drawn from those three sets by this argument reaches ours. Curves of
  constant width were proposed for this role, and explored numerically, by Gibbs
  in 2014; what is added here is a proof. It consists of an analytic reduction
  followed by one finite computation. The reduction bounds the hull area from
  below over an entire box of placements at once, by eroding each Reuleaux polygon
  to a fixed set contained in every placement that box allows. The computation is
  an exhaustive subdivision of the resulting five-dimensional space, recorded as a
  certificate of 486,799,600 nodes and checked by a verifier independent of the
  search, with a rigorous bound on its floating point error some
  thousands of times smaller than the margin the verification attains.
\end{abstract}

\section{Introduction}
\label{sec:intro}

A planar set has diameter one when no two of its points lie farther than one
apart. Many sets have this property and they look very unlike one another. The
disc of diameter one is one of them. So is the equilateral triangle of side one.
So is the Reuleaux triangle, obtained from that triangle by replacing each side
with a circular arc centred at the opposite vertex. A convex set in the plane is
a \emph{universal cover} if every set of diameter one can be moved to fit inside
it, by a translation, a rotation, and a reflection where a reflection helps.
Lebesgue asked, in a 1914 letter to P\'al, how small the area of a universal
cover can be. The answer is not known, and the problem is recorded as open by
Brass, Moser and Pach~\cite{BrassMoserPach2005}. Convexity is part of the
question rather than a convenience: Duff~\cite{Duff1980} constructs a non-convex
universal cover of smaller area than any convex one then known, so the convex and
unrestricted variants have different answers.

The difficulty is that a single set has to serve all of them at once, and none of
them will serve as a cover for the others. Take the disc of diameter one and the
equilateral triangle of side one. Neither contains the other. The triangle's
corners lie at distance $1/\sqrt3$ from its centre, beyond the disc's radius of
$1/2$, so the triangle does not fit in the disc however it is turned; and the
largest disc that fits inside the triangle has radius $1/(2\sqrt3)$, so the disc
does not fit in the triangle either. A cover must find room for both, and for every other
set of diameter one, each in some position of its own. Write $\Lambda$ for the
least area of a universal cover, which exists by the Blaschke selection theorem.

Progress on $\Lambda$ has come from narrowing the interval that contains it, and
the two sides of that interval have quite different characters.

Upper bounds are constructive. One exhibits a cover and then removes as much of
it as can be spared. P\'al~\cite{Pal1920} starts from the regular hexagon of inradius
$1/2$, which is a universal cover, and cuts away two of its corners.
Sprague~\cite{Sprague1936} removes a further piece.
Hansen~\cite{Hansen1975} reduces it again, and later removes two very thin
slivers~\cite{Hansen1992}. The current record,
\GibbsUpper{}, is due to Gibbs~\cite{Gibbs2018}, and refines the construction
that Baez, Bagdasaryan and Gibbs~\cite{BaezBagdasaryanGibbs2015} analysed to high
precision. Each step in this line is a smaller and more delicate excision than
the one before it.

Lower bounds rest on one observation, elementary enough to state in a sentence.
Fix finitely many sets of diameter one. A universal cover contains a congruent
copy of each of them, and being convex it also contains the convex hull of the
union of those copies. The copies may sit anywhere, so the area of the cover is
at least the smallest area that hull can have, minimised over every way of
placing them. Choosing the sets, and carrying out that minimisation, gives a
bound.

The bounds obtained this way have moved much less than the upper bounds. P\'al~\cite{Pal1920}
takes a disc and an equilateral triangle, evaluates the minimum by hand, and
obtains $\PalLB$. Elekes~\cite{Elekes1994} takes a disc together with regular
$3^i$-gons and reaches approximately \ElekesLB{}. Brass and
Sharifi~\cite{BrassSharifi2005} take a disc, an equilateral triangle and a
regular pentagon; the placements of three bodies form a five-parameter family
once the redundant degrees of freedom are removed, and they minimise over it by
subdivision, obtaining
\BSLB{}. That bound has stood since 2005 and is still the best that has been
refereed. A 2026 preprint of Xie~\cite{Xie2026} reports \XieLB{} from the same three
sets.

Every one of these bounds comes from a disc together with regular polygons, and
the last two come from the same three sets, twenty-one years apart. That pattern
suggests a configuration close to exhausted, and we make the suggestion precise
rather than leaving it as an impression. The
disc, equilateral triangle and regular pentagon admit a placement whose hull has
area at most \ClassicalCeiling{} (\cref{prop:ceiling}), and a lower bound drawn
from a family can never exceed the minimum hull area of that family. So no bound
drawn from those three sets through \cref{lem:testset}, however long the
computation runs, can exceed \ClassicalCeiling{}.

Our result passes that ceiling.

\begin{theorem}
\label{thm:main}
Every convex universal cover for the planar sets of diameter one has area at
least \MainBound{}. That is, $\Lambda \ge \MainBound$.
\end{theorem}

The proof separates into an analytic part and a single finite computation, and we
keep the two apart throughout.
\Cref{sec:prelim,sec:ceiling,sec:reduction} can be read without a computer. They
fix the test sets, prove the ceiling just quoted, confine the placements to a
compact box, and establish an estimate that holds simultaneously for every
placement in a box, which is what makes a finite search possible at all.
\Cref{sec:computation} carries out the search, records it as a certificate, and
describes the independent verifier that checks it.

The change that gets past the ceiling is a change of test sets, and it costs
nothing in the size of the computation. A body of constant width one has diameter
one, so bodies of constant width are admissible test sets. The Reuleaux triangle
has constant width one and contains the equilateral triangle of side one, since
it is that triangle with three arcs added. The Reuleaux pentagon stands in the
same relation to the regular pentagon of diameter one. Replacing each polygon by
the Reuleaux polygon built on it therefore enlarges every copy the cover must
contain, while leaving the number of bodies, the number of placement parameters,
and the structure of the search exactly as they were. We use the disc, the
Reuleaux triangle and the Reuleaux pentagon, all of width one.

\subsection{Provenance of the test sets}
\label{sec:contributions}

That choice of test sets is not ours, and the attribution is worth stating
precisely.
A set of diameter one is contained in a curve of constant width
one~\cite{Grunbaum1963}, so it is enough to consider bodies of constant width,
and among those the Reuleaux polygons are the easiest to construct.
Gibbs~\cite{Gibbs2014} states both points, observes that Reuleaux polygons appear
empirically to be the most effective choice for a given number of test sets, and
runs a simulated annealing search on a family of five of them, reporting
\GibbsHeuristic{}. He is careful about what that number is, calling it an upper
bound on such a lower bound rather than a bound.

That distinction is not a formality, and it is the reason this paper exists.
Write $M$ for the smallest hull area attainable by a chosen family of test sets,
over all ways of placing them, so that the observation above reads
$\Lambda\ge M$. A search of any kind
exhibits a placement, and the area of an exhibited placement is an upper bound on
$M(S)$. The quantity that has to be bounded from below is therefore bounded from
above, and the direction is wrong by construction rather than by accident. A
better optimiser does not repair it. Our own optimiser, run on the same family of
five bodies, returns \GibbsOursOnHis{}, and the placement Gibbs exhibits in 2014
already shows $M$ to be smaller than that, so our value would have been
unsound to publish as a bound.

What this paper contributes is the missing direction, together with three things
needed to obtain it.

\begin{enumerate}
\item An estimate that bounds the hull area from below simultaneously for every
      placement in a box of placements (\cref{lem:erosion}). Certified bounds of
      this kind have been obtained for the classical polygonal test sets by
      Xie~\cite{Xie2026}, who carries the witness points through a parameter
      domain by interval arithmetic. \Cref{lem:erosion} instead exhibits a fixed
      body inside every placement the box allows, which is what makes an
      exhaustive search possible for test sets bounded by circular arcs. Gibbs
      replaces his arcs by fine polygons; we do not approximate the bodies.
\item The bound of \cref{thm:main}, exhaustive over the whole placement space and
      recorded as a certificate of \Nodes{} nodes. Its verification is logically
      independent of the search: the verifier reads the certificate alone and
      re-derives every estimate it checks from the definitions, so the
      correctness of the search procedure is not part of the argument.
\item An exact ceiling for the classical family (\cref{prop:ceiling}), which
      shows the improvement to be a change of test sets rather than a longer
      search.
\end{enumerate}

\subsection{Related work on computer-assisted proofs}
\label{sec:relatedcomp}

Proofs that reduce to a finite computation over a subdivided parameter space have
a long history in discrete geometry. The four colour theorem is the standard
reference point, proved by Appel and Haken~\cite{AppelHaken1977a,AppelHaken1977b}
and reproved with a smaller machine-checked case analysis by Robertson, Sanders,
Seymour and Thomas~\cite{RobertsonEtAl1997}. The formal proof of the Kepler
conjecture given by the authors of~\cite{HalesEtAl2017} combines combinatorial
reduction with interval arithmetic over a subdivision, in a pattern close to the
one used here. Szekeres and Peters~\cite{SzekeresPeters2006} settle the
seventeen-point Erd\H{o}s--Szekeres problem by an exhaustive computation, and Hass
and Schlafly~\cite{HassSchlafly2000} prove the double bubble conjecture in a
symmetric case by a computer-assisted argument. Nearer to the present problem,
Fekete, Gurunathan, Juneja, Keldenich, Kleist and
Scheffer~\cite{FeketeEtAl2021} determine the critical density for packing squares
into a disc with a manual analysis completed by interval arithmetic over
subdivided cuboids.

Our computation differs from those in one respect worth stating early. We do not
use interval arithmetic. The quantity checked at each leaf is the area of the
convex hull of finitely many points, each of which is tested for membership in
the body it is supposed to lie in, so only two things can perturb it: a point
accepted slightly outside its body, and the rounding of the shoelace sum. Both
admit elementary bounds, given in \cref{sec:errors}, and their sum is \Eta.

\section{Preliminaries}
\label{sec:prelim}

\subsection{Test sets}

Everything in this paper rests on one lemma, and the lemma is short. A universal
cover has to contain a copy of each of our chosen sets. It is convex, so it also
contains the convex hull of those copies taken together. We do not get to choose
where the copies sit, so we must take the worst arrangement, and even the worst
arrangement gives a bound.

\begin{definition}
\label{def:cover}
The \emph{diameter} of a bounded set $T\subset\mathbb{R}^2$ is
$\operatorname{diam} T=\sup\{\,|x-y| : x,y\in T\,\}$. A convex set
$K\subset\mathbb{R}^2$ is a \emph{universal cover} if for every $T$ with
$\operatorname{diam} T\le 1$ there is a rigid motion $g$ of the plane, orientation
reversing if required, with $g(T)\subseteq K$. We write $\Lambda$ for the infimum
of $\operatorname{area}(K)$ over all universal covers $K$.
\end{definition}

\begin{definition}
\label{def:M}
For sets $T_1,\dots,T_k$ of diameter at most one, put
\[
  M(T_1,\dots,T_k)
  \;=\;
  \inf_{g_1,\dots,g_k}\;
  \operatorname{area}\Bigl(\operatorname{conv}\bigl(\textstyle\bigcup_{i}g_i(T_i)\bigr)\Bigr),
\]
the infimum running over all $k$-tuples of rigid motions of the plane. Where the
family is clear we write simply $M$.
\end{definition}

\begin{lemma}[Test sets]
\label{lem:testset}
For any sets $T_1,\dots,T_k$ of diameter at most one,
$\Lambda \ge M(T_1,\dots,T_k)$.
\end{lemma}

\begin{proof}
Let $K$ be a universal cover. Each $T_i$ has diameter at most one, so by
\cref{def:cover} there is a rigid motion $g_i$ with $g_i(T_i)\subseteq K$. Then
$\bigcup_i g_i(T_i)\subseteq K$, and $K$ is convex, so
$\operatorname{conv}(\bigcup_i g_i(T_i))\subseteq K$. Taking areas,
\[
  \operatorname{area}(K)
  \;\ge\;
  \operatorname{area}\Bigl(\operatorname{conv}\bigl(\textstyle\bigcup_i g_i(T_i)\bigr)\Bigr)
  \;\ge\; M(T_1,\dots,T_k),
\]
the second inequality because the particular motions $g_i$ compete in the
infimum of \cref{def:M}. The bound holds for every universal cover $K$, hence for
the infimum of their areas.
\end{proof}

The two directions are not symmetric. A lower bound on $M$ controls every tuple
of motions at once; an upper bound needs only one, since any single tuple competes
in the infimum. \Cref{sec:ceiling} uses the easy direction and
\cref{sec:reduction,sec:computation} the hard one.

\subsection{Constant width}

Larger test sets give larger hulls, so we want each $T_i$ to be as large as a
diameter-one set can be. Width is the right way to see which sets those are. The
width of a convex body in a given direction is the distance between its two
supporting lines perpendicular to that direction, and the diameter turns out to
be the largest width over all directions. A body whose width is the same in every
direction therefore has its diameter equal to that common width, which is as
large as the body can be relative to how far it extends.

For a convex body $K$ write $h_K(\theta)=\max_{x\in K}\langle x,u(\theta)\rangle$
for its support function, where $u(\theta)=(\cos\theta,\sin\theta)$, and
$w_K(\theta)=h_K(\theta)+h_K(\theta+\pi)$ for its width in direction $\theta$.
$K$ has \emph{constant width} $w$ if $w_K\equiv w$.

\begin{lemma}
\label{lem:width}
A convex body $K$ satisfies
$\operatorname{diam} K=\max_{\theta}w_K(\theta)$. In particular a body of constant
width $w$ has diameter $w$.
\end{lemma}

\begin{proof}
For $x,y\in K$ with $x\neq y$ set $u=(x-y)/|x-y|$ and let $\theta$ be its angle.
Then $|x-y|=\langle x,u\rangle-\langle y,u\rangle\le
h_K(\theta)+h_K(\theta+\pi)=w_K(\theta)$, so
$\operatorname{diam} K\le\max_\theta w_K(\theta)$. Conversely fix $\theta$ and
choose $x,y\in K$ with $\langle x,u(\theta)\rangle=h_K(\theta)$ and
$\langle y,u(\theta+\pi)\rangle=h_K(\theta+\pi)$, which exist because $K$ is
compact. Then
$|x-y|\ge\langle x-y,u(\theta)\rangle=w_K(\theta)$, so
$\operatorname{diam} K\ge w_K(\theta)$ for every $\theta$.
\end{proof}

\begin{remark}
\label{rem:completion}
Every set of diameter at most one is contained in a body of constant width
one~\cite{Grunbaum1963}, so the bodies of constant width one are exactly the
maximal test sets and nothing is lost by restricting to them. Martini, Montejano
and Oliveros~\cite{MartiniMontejanoOliveros2019} give a comprehensive account of
such bodies, the completions among them included. We record this
because it is why constant width is the right family to look in. The proof of
\cref{thm:main} does not use it. All that is required there is
\cref{lem:width}, which certifies that the particular bodies we choose have
diameter one and are therefore admissible in \cref{lem:testset}.
\end{remark}

\subsection{Reuleaux polygons}

The Reuleaux polygons are the constant-width bodies that are easiest to write
down. Take an odd number of points, each pair at distance at most one, and
intersect the unit discs centred at them. For the vertices of a regular polygon of
diameter one this produces the familiar rounded shape bounded by circular arcs,
and it contains the polygon it was built from.
\Cref{fig:testsets} shows the three bodies we use and the containment.

\begin{figure}[t]
\centering
\includegraphics{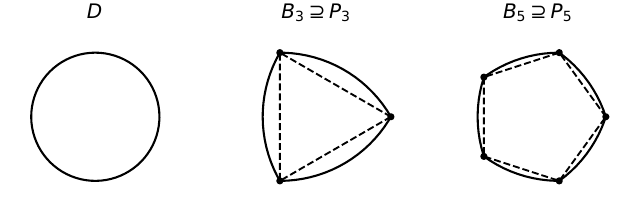}
\caption{The test sets of \cref{def:family}, drawn to a common scale. Left, the
disc $D$ of diameter one. Centre and right, the Reuleaux polygons $B_3$ and $B_5$
of width one (solid) with the regular polygons $P_3$ and $P_5$ of diameter one on
the same corners (dashed), the corners marked. The containment is
\cref{lem:contains}.}
\label{fig:testsets}
\end{figure}

\begin{definition}
\label{def:reuleaux}
Let $n\ge3$ be odd and let $V_1,\dots,V_n$ be the vertices of a regular $n$-gon of
diameter one, so $|V_i-V_j|\le1$ for all $i,j$. The \emph{regular Reuleaux
$n$-gon} of width one is
$B_n=\bigcap_{j=1}^{n}\overline{D}(V_j,1)$, where $\overline{D}(c,r)$ is the
closed disc of centre $c$ and radius $r$. We call $V_1,\dots,V_n$ its
\emph{corners}, and write $R_n=|V_j|$ for their common distance from the centroid,
so that $R_3=\CircThree$ and $R_5=\CircFive$.
\end{definition}

$B_n$ has constant width one, which is classical~\cite{Grunbaum1963,KellyWeiss1979}.
By \cref{lem:width} it has diameter one, so it is admissible in
\cref{lem:testset}.

\begin{lemma}
\label{lem:contains}
$B_n$ contains the regular $n$-gon $P_n=\operatorname{conv}\{V_1,\dots,V_n\}$ of
diameter one on the same corners.
\end{lemma}

\begin{proof}
Fix $i$. For every $j$ we have $|V_i-V_j|\le1$, so $V_i\in\overline{D}(V_j,1)$,
and therefore $V_i\in B_n$. Thus $B_n$ contains every corner, and $B_n$ is convex
as an intersection of discs, so it contains their convex hull $P_n$.
\end{proof}

\Cref{lem:contains} is the whole of the improvement, stated once. Enlarging a
test set can only enlarge the hull of any arrangement containing it, so replacing
$P_n$ by $B_n$ increases $M$, or leaves it unchanged, for every family in which
$P_n$ appears. It does so without changing the number of test sets or the number
of parameters describing their placements.

\begin{definition}
\label{def:family}
Our test sets are $D=\overline{D}(0,\tfrac12)$, the closed disc of diameter one
centred at the origin, together with $B_3$ and $B_5$. The classical family is $D$ together with $P_3$ and $P_5$, the equilateral
triangle of side one and the regular pentagon of diameter one.
\end{definition}

\subsection{Support functions and hulls}

Hulls of unions are computed through support functions in
\cref{sec:ceiling,sec:computation}, and one identity carries all of it.

\begin{lemma}
\label{lem:support}
For compact convex sets $K_1,\dots,K_m$,
$h_{\operatorname{conv}(\bigcup_i K_i)}(\theta)=\max_i h_{K_i}(\theta)$ for every
$\theta$. If $K$ is translated by $t$ and rotated by $\rho$ about the origin then
its support function becomes $\theta\mapsto h_K(\theta-\rho)+\langle
t,u(\theta)\rangle$.
\end{lemma}

\begin{proof}
A linear functional attains its maximum over $\operatorname{conv}(\bigcup_i K_i)$
at an extreme point of that hull, and every extreme point lies in some $K_i$, so
the maximum equals $\max_i h_{K_i}(\theta)$. The second statement is immediate
from the definition.
\end{proof}

\section{The classical family and its ceiling}
\label{sec:ceiling}

Bounding $M$ from above takes no search. Any single arrangement of the test sets
competes in the infimum of \cref{def:M}, so exhibiting one arrangement and
measuring its hull settles the matter. We therefore overestimate the hull
deliberately, by intersecting the half-planes that support
it in finitely many directions. That region contains the hull whatever the
directions are, so its area is an upper bound with no approximation argument
attached.

\begin{lemma}
\label{lem:outer}
Let $K$ be a compact convex set and let $\theta_1,\dots,\theta_N$ be any
directions. Then
\[
  K \;\subseteq\; Q \;=\;
  \bigl\{\,x\in\mathbb{R}^2 : \langle x,u(\theta_j)\rangle\le h_K(\theta_j)
  \text{ for } j=1,\dots,N \,\bigr\},
\]
and hence $\operatorname{area}(K)\le\operatorname{area}(Q)$.
\end{lemma}

\begin{proof}
If $x\in K$ then $\langle x,u(\theta_j)\rangle\le h_K(\theta_j)$ for every $j$, by
the definition of the support function. So $x\in Q$. Area is monotone under
inclusion.
\end{proof}

Taking the $\theta_j$ equally spaced makes $Q$ a convex polygon whose vertices are
the intersections of consecutive constraint lines, and its area follows from the
shoelace formula. Increasing $N$ shrinks $Q$, so the bound improves, and every
value obtained along the way remains valid.

\begin{proposition}
\label{prop:ceiling}
$M(D,P_3,P_5) \le \ClassicalCeilingSafe$.
\end{proposition}

\begin{proof}
Place $P_3$ by the translation $\ClassicalTthree$ and $P_5$ by the rotation
$\ClassicalRho$ followed by the translation $\ClassicalTfive$, leaving $D$
centred at the origin. By
\cref{lem:support} the support function of the hull of the union is the pointwise
maximum of the three support functions, each of which is elementary. Applying
\cref{lem:outer} with $N=\CeilingDirs$ equally spaced directions gives an outer
polygon of area at most \ClassicalOuterC{}, the constants here and below being
rounded upwards from the computed values. This arrangement competes in the
infimum of \cref{def:M}, so
$M(D,P_3,P_5)\le\ClassicalOuterC<\ClassicalCeilingSafe$.
\end{proof}

\begin{remark}
\label{rem:provenance}
How the arrangement was found does not enter the proof. We located it by
numerical minimisation, but any arrangement would give a valid bound, and the one
quoted is checked directly. For the same reason the constant is stated with room
to spare: the outer polygon areas at $N=10^4$, $10^5$ and $\CeilingDirs$ are
\ClassicalOuterA{}, \ClassicalOuterB{} and \ClassicalOuterC{}, a decreasing
sequence of valid upper bounds, and rounding to \ClassicalCeilingSafe{} puts the
statement far above any question about the arithmetic.
\end{remark}

\begin{corollary}
\label{cor:exhausted}
Every lower bound on $\Lambda$ obtained from the test sets $D$, $P_3$ and $P_5$
through \cref{lem:testset} is at most \ClassicalCeilingSafe{}, and in particular
smaller than \MainBound{}.
\end{corollary}

\begin{proof}
\Cref{lem:testset} applied to these three sets yields the single bound
$\Lambda\ge M(D,P_3,P_5)$, and \cref{prop:ceiling} bounds that quantity.
\end{proof}

These are the test sets of Brass and Sharifi~\cite{BrassSharifi2005} and of
Xie~\cite{Xie2026}, and P\'al's pair is contained in them.
\Cref{cor:exhausted} therefore settles what the classical configuration can
deliver. The published values \BSLB{} and \XieLB{} are already within
$1.6\times10^{-3}$ and $6\times10^{-4}$ of its ceiling, and no amount of further
computation on those sets reaches \MainBound{}. Passing that ceiling requires
different test sets, which is what \cref{sec:reduction,sec:computation} do.

The same computation applied to our own test sets gives
$M(D,B_3,B_5)\le\ReuleauxOuterC$, from the arrangement
\[
  \ReuleauxWitness.
\]
That number is not needed for \cref{thm:main}. It records how much the family we
do use still has in reserve, and \cref{sec:scope} returns to the cost of
extracting it.

\section{Reduction to a finite computation}
\label{sec:reduction}

\Cref{lem:testset} turns \cref{thm:main} into a statement about
$M(D,B_3,B_5)$, an infimum over an unbounded family of triples of rigid motions.
Three steps make that finite. Symmetry removes the redundant parameters, an area
estimate confines the rest to a compact box, and an erosion estimate bounds the
hull area from below over a whole box of placements at once. Only the third is
specific to this problem.

\subsection{The gauge}

The area of a hull does not change when the whole arrangement is moved, so some
of the nine parameters describing three rigid motions carry no information. A
global translation places $D$ at the origin, and $D$ is invariant under rotation
about its centre, so $D$ contributes nothing. A global rotation then fixes the
orientation of $B_3$. What remains is the translation of $B_3$, and the rotation
and translation of $B_5$. We call this normalisation \emph{fixing the gauge},
and it is what reduces nine parameters to five.

\begin{definition}
\label{def:placement}
A \emph{placement} is a triple $(t_3,\rho,t_5)\in\mathbb{R}^2\times\mathbb{R}
\times\mathbb{R}^2$, standing for the arrangement of $D$ centred at the origin,
$B_3$ translated by $t_3$, and $B_5$ rotated by $\rho$ about its centre and then
translated by $t_5$. Here $\Rot_\rho$ denotes the rotation of the plane through
angle $\rho$ about the origin, written this way so that it is not confused with
the circumradius $R_n$ of \cref{def:reuleaux}. Write
\[
  A(t_3,\rho,t_5)=\operatorname{area}
  \bigl(\operatorname{conv}(D\cup(B_3+t_3)\cup(\Rot_\rho B_5+t_5))\bigr).
\]
\end{definition}

\begin{definition}
\label{def:box}
A \emph{box of placements} is a product of five closed intervals, one for each
coordinate of \cref{def:placement}, written through its centre and half-widths as
\[
  \mathcal{B}
  =\prod\bigl[c-h,\;c+h\bigr]
  \quad\text{over the coordinates } t_{3x},t_{3y},\rho,t_{5x},t_{5y},
\]
with half-widths $h_{3x},h_{3y},h_\rho,h_{5x},h_{5y}\ge0$. We abbreviate
$h_x,h_y$ for the two translation half-widths of whichever body is under
discussion.
\end{definition}

\begin{lemma}
\label{lem:gauge}
$M(D,B_3,B_5)=\inf\{A(t_3,\rho,t_5)\}$, the infimum over
$t_3,t_5\in\mathbb{R}^2$ and $\rho\in[0,2\pi/5)$.
\end{lemma}

\begin{proof}
Given any triple of rigid motions, compose all three with one further rigid
motion carrying the image of $D$ to the disc centred at the origin and the image
of $B_3$ to a prescribed orientation. This changes no area, and $D$ is invariant
under rotations about its centre, so the resulting arrangement is a placement in
the sense of \cref{def:placement}. Reflections need not be considered separately,
since $D$, $B_3$ and $B_5$ each have an axis of symmetry. Rotating $B_5$ about
its own centre by $2\pi/5$ carries it to itself, so $\rho$ may be restricted to
$[0,2\pi/5)$.
\end{proof}

\subsection{Confining the placements}

A test body sitting far from the disc forces a large hull on its own, whatever
the rest of the arrangement does, because the hull already contains the disc and
a distant point of that body. The area of the convex hull of a disc and a single
point is elementary, and it grows with the distance, so beyond some distance the
arrangement cannot compete with the bound we are trying to prove and need not be
examined.

\begin{lemma}
\label{lem:apriori}
For $d\ge1/2$ let
\[
  f(d)=\tfrac14\bigl(\pi-\arccos\tfrac{1}{2d}\bigr)+\tfrac12\sqrt{d^{2}-\tfrac14}
\]
be the area of the convex hull of $D$ and a point at distance $d$ from its
centre. Then $f$ is strictly increasing, and if any point of the arrangement lies
at distance $d\ge 1/2$ from the origin then $A(t_3,\rho,t_5)\ge f(d)$.
\end{lemma}

\begin{proof}
The hull of $D$ and a point $p$ at distance $d$ consists of the two tangent
segments from $p$, together with the part of the disc they cut off. Writing
$r=1/2$, the kite bounded by the two tangent segments and the two radii to the
points of tangency has area $r\sqrt{d^{2}-r^{2}}$, and the remaining circular
sector subtends $2\pi-2\arccos(r/d)$ and has area $r^{2}(\pi-\arccos(r/d))$.
Adding these and putting $r=1/2$ gives $f$, which increases in $d$ because both
terms do. The arrangement contains $D$ and the point $p$, and it is convex, so
its area is at least $f(d)$.
\end{proof}

\begin{corollary}
\label{cor:domain}
If $A(t_3,\rho,t_5)<\MainBound$ then $|t_3|\le\TmaxNaive$ and
$|t_5|\le\TmaxNaive$.
\end{corollary}

\begin{proof}
The centre of $B_3$ lies in $B_3$, so after translation it lies at distance
$|t_3|$ from the origin. If $|t_3|>\TmaxNaive$ then \cref{lem:apriori} gives
$A\ge f(|t_3|)>f(\TmaxNaive)>\fAtTmax>\MainBound$. The same argument applies to
$t_5$.
\end{proof}

The search therefore runs over the compact box
$[-\TmaxNaive,\TmaxNaive]^{2}\times[0,2\pi/5]\times[-\TmaxNaive,\TmaxNaive]^{2}$.

\begin{remark}
\label{rem:sharperdomain}
\Cref{cor:domain} is not the sharpest available. Applying \cref{lem:apriori} to
the farthest corner of a body rather than to its centre uses the fact that the
corners already lie at distance $R_n$ from the centre before any translation. For
$n$ corners equally spaced, some corner lies within angle $\pi/n$ of any given
direction, so $\max_j|\Rot_\rho V_j+t|\ge\bigl(R_n^{2}+|t|^{2}+2R_n|t|\cos(\pi/n)
\bigr)^{1/2}$ uniformly in $\rho$, and the resulting domain is
$|t_3|\le\TmaxThree$ and $|t_5|\le\TmaxFive$, smaller by a factor of
\DomainShrink{} in four-dimensional volume. This recovers the box used by Brass
and Sharifi~\cite{BrassSharifi2005}. It does not reduce the cost of the
computation, for the reason given in \cref{rem:nosaving}, and the certificate on
record uses \cref{cor:domain}.
\end{remark}

\subsection{The erosion estimate}

The difficulty in searching a box of placements is that each box contains
infinitely many of them. Subdividing helps only if a single computation can
settle an entire box, and for that we need one body that is contained in every
placement the box allows. A Reuleaux polygon is an intersection of unit discs
about its corners, and moving the body moves those corners. If every corner moves
by at most $\delta$, then shrinking every disc by $\delta$ before moving produces
a body inside all of them at once, because a point within $1-\delta$ of the
original centre is within $1$ of the moved one. That is the triangle inequality
and nothing more. \Cref{fig:erosion} draws the situation for $B_3$.

\begin{figure}[t]
\centering
\includegraphics{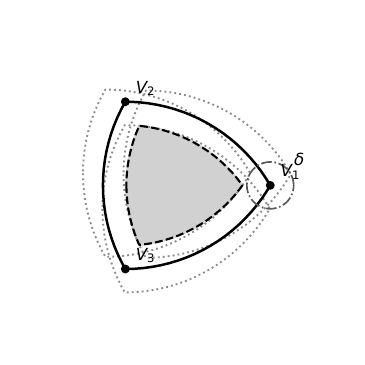}
\caption{The estimate of \cref{lem:erosion} for $B_3$. Solid, the body
$\bigcap_j\overline{D}(V_j,1)$ with its corners at the centre of a box. Dotted,
three of the placements that box allows, each moving every corner by at most
$\delta$, the dash-dotted circle showing that budget at $V_1$. Shaded, the set
$\bigcap_j\overline{D}(V_j,1-\delta)$, which lies inside every one of them. Drawn
with $\delta=0.14$ for legibility; the certificate reaches values many orders
smaller.}
\label{fig:erosion}
\end{figure}

\begin{lemma}[Erosion]
\label{lem:erosion}
Let $V_1,\dots,V_n$ and $W_1,\dots,W_n$ be points with $|W_j-V_j|\le\delta$ for
every $j$, where $0\le\delta\le1$. Then
\[
  \bigcap_{j=1}^{n}\overline{D}(V_j,1-\delta)
  \;\subseteq\;
  \bigcap_{j=1}^{n}\overline{D}(W_j,1).
\]
The left-hand side is non-empty precisely when $1-\delta$ is at least the
circumradius of $\{V_1,\dots,V_n\}$, that is, at least
$\min_p\max_j|p-V_j|$.
\end{lemma}

\begin{proof}
Let $p$ belong to the left-hand side and fix $j$. Then
$|p-W_j|\le|p-V_j|+|V_j-W_j|\le(1-\delta)+\delta=1$, so
$p\in\overline{D}(W_j,1)$. As $j$ was arbitrary, $p$ lies in the right-hand side.
For the second statement, $\bigcap_j\overline{D}(V_j,\varrho)$ is non-empty if and
only if some $p$ satisfies $|p-V_j|\le\varrho$ for all $j$, which is to say
$\min_p\max_j|p-V_j|\le\varrho$.
\end{proof}

We apply this with $V_j$ the corners of a test body positioned at the centre of a
box, and $W_j$ its corners at an arbitrary placement in that box. What is needed
is a bound on how far a corner can move.

\begin{lemma}
\label{lem:displacement}
Let a box of placements in the sense of \cref{def:box} have translation
half-widths $h_x,h_y$ and rotation half-width $h_\rho$. A corner of $B_3$, whose orientation is fixed
by the gauge, moves by at most $\delta_3=\sqrt{h_x^{2}+h_y^{2}}$ from its
position at the centre of the box. A corner of $B_5$ moves by at most
$\delta_5=\sqrt{h_x^{2}+h_y^{2}}+2R_5\sin(h_\rho/2)$.
\end{lemma}

\begin{proof}
A translation differing from the box centre by $(\Delta_x,\Delta_y)$ displaces
every point by at most $\sqrt{h_x^{2}+h_y^{2}}$. A rotation by $\Delta_\rho$
carries a point at distance $R_5$ from the body centre through a chord of length
$2R_5\sin(|\Delta_\rho|/2)\le2R_5\sin(h_\rho/2)$. The triangle inequality bounds
the composition by the sum.
\end{proof}

\subsection{The bound on a box}

Combining the last three statements gives a single quantity that bounds $A$ from
below throughout a box, which is what the search evaluates.

\begin{proposition}
\label{prop:boxbound}
Let $\mathcal{B}$ be a box of placements as in \cref{def:box}, let $V^{(3)}_j$
and $V^{(5)}_j$ be the
corners of $B_3$ and $B_5$ positioned at its centre, and let $\delta_3,\delta_5$
be as in \cref{lem:displacement}. Put
$C_n=\bigcap_j\overline{D}(V^{(n)}_j,1-\delta_n)$ and
\[
  \beta(\mathcal{B})=\operatorname{area}
  \bigl(\operatorname{conv}(D\cup C_3\cup C_5)\bigr).
\]
Then $A(t_3,\rho,t_5)\ge\beta(\mathcal{B})$ for every placement in $\mathcal{B}$.
\end{proposition}

\begin{proof}
Fix a placement in $\mathcal{B}$ and let $W^{(n)}_j$ be the corners of the two
bodies as placed. By \cref{lem:displacement}, $|W^{(n)}_j-V^{(n)}_j|\le\delta_n$
for every $j$, so \cref{lem:erosion} gives $C_n\subseteq\bigcap_j
\overline{D}(W^{(n)}_j,1)$, which is the placed body. Hence
$D\cup C_3\cup C_5$ is contained in the union of the three placed bodies, and
taking convex hulls and then areas preserves the inclusion.
\end{proof}

\begin{remark}
\label{rem:degenerate}
When $\delta_n$ is large enough that $C_n$ is empty, which by \cref{lem:erosion}
happens exactly when $1-\delta_n$ falls below $R_n$, the proposition still holds
with the empty set contributing nothing, and the search falls back on
\cref{lem:apriori}. The thresholds are $\delta_3=\InradThree$ and
$\delta_5=\InradFive$. Boxes that large occur only at the top of the subdivision.
\end{remark}

\begin{remark}
\label{rem:notgrid}
\Cref{prop:boxbound} is an inequality about every placement in $\mathcal{B}$, not
a sample of them. The subdivision that follows decides only where to stop
refining, and no step of the argument treats a finite set of placements as
representative of the box that contains it.
\end{remark}

\section{The certified computation}
\label{sec:computation}

\Cref{prop:boxbound} settles a whole box at once, so the remaining task is to
cover the domain of \cref{cor:domain} by finitely many boxes on each of which the
bound is good enough. That is a subdivision, and it either finishes or it does
not. This section describes the subdivision, the record it leaves behind, and the
separate program that checks the record.

Throughout, $\tau=\MainBound$ denotes the value being certified and
$\mathcal{D}=[-\TmaxNaive,\TmaxNaive]^{2}\times[0,2\pi/5]\times
[-\TmaxNaive,\TmaxNaive]^{2}$ the domain.

\subsection{The subdivision}

\begin{proposition}
\label{prop:cover}
Suppose $\mathcal{B}_1,\dots,\mathcal{B}_N$ are boxes of placements with
$\mathcal{D}\subseteq\bigcup_i\mathcal{B}_i$ and
$\beta(\mathcal{B}_i)\ge\tau$ for every $i$. If in addition
$\tau\le f(\TmaxNaive)$, then $\Lambda\ge\tau$.
\end{proposition}

\begin{proof}
Let $(t_3,\rho,t_5)$ be any placement. By \cref{lem:gauge} we may take
$\rho\in[0,2\pi/5)$, so only the translations can leave $\mathcal{D}$. If the
placement lies outside it then $|t_3|>\TmaxNaive$ or $|t_5|>\TmaxNaive$, and
\cref{lem:apriori} gives $A\ge f(\TmaxNaive)\ge\tau$. Otherwise it lies in some $\mathcal{B}_i$, and
\cref{prop:boxbound} gives $A\ge\beta(\mathcal{B}_i)\ge\tau$. So $A\ge\tau$
everywhere, hence $M(D,B_3,B_5)\ge\tau$ by \cref{lem:gauge}, and
\cref{lem:testset} gives $\Lambda\ge\tau$.
\end{proof}

The covering is produced by bisection. A box whose bound already reaches $\tau$
is discarded, since \cref{prop:cover} asks nothing more of it. Any other box is
split into two, and the two are treated in the same way. Which coordinate to
split matters for efficiency and not for correctness: we split the one on which
the box is widest, measuring the rotation coordinate by the distance a corner of
$B_5$ travels, so that $h_\rho$ is compared as $R_5h_\rho$ against the
translation half-widths.

\begin{definition}
\label{def:split}
The \emph{split} of a box $\mathcal{B}$ halves its widest coordinate in the
weighting just described, producing two boxes whose centres are displaced by
$h/2$ and whose half-width in that coordinate is $\tfrac{h}{2}(1+\varepsilon)$
with $\varepsilon=\Inflate$, the other coordinates unchanged.
\end{definition}

\begin{lemma}
\label{lem:covers}
The two boxes produced by \cref{def:split} cover $\mathcal{B}$.
\end{lemma}

\begin{proof}
Write the split coordinate as $[c-h,c+h]$. The children carry that coordinate
over the intervals centred at $c\pm h/2$ of half-width
$\tfrac{h}{2}(1+\varepsilon)$, namely
$[c-h-\tfrac{h\varepsilon}{2},\,c+\tfrac{h\varepsilon}{2}]$ and
$[c-\tfrac{h\varepsilon}{2},\,c+h+\tfrac{h\varepsilon}{2}]$. Their union
contains $[c-h,c+h]$, and every other coordinate is unchanged.
\end{proof}

The widening by $\varepsilon$ makes the covering hold with room to spare rather
than exactly, so that the conclusion does not rest on two intervals meeting at a
single shared endpoint.

\begin{corollary}
\label{cor:leavescover}
The boxes at which a run stops cover the box it started from.
\end{corollary}

\begin{proof}
Induction on the subdivision. A box that is discarded covers itself. A box that
is split is covered by its two children by \cref{lem:covers}, and each child is
covered by the boxes at which its own subtree stops, by the inductive hypothesis.
\end{proof}

\Cref{cor:leavescover} supplies the covering hypothesis of \cref{prop:cover}, so
what remains to be checked is only that every box at which the run stopped
satisfies $\beta\ge\tau$.

\subsection{The certificate}

Storing the boxes examined is unnecessary. The split of \cref{def:split} is a
deterministic function of a box, so the subdivision is recovered from its shape
alone, and that shape is a binary tree.

\begin{definition}
\label{def:certificate}
The \emph{search tree} of a run has the domain box at its root; a
\emph{node} is a box examined during the run; a node is a \emph{leaf} if it was
discarded, and otherwise has exactly the two children given by
\cref{def:split}. The \emph{certificate} is the sequence of bits obtained by
traversing the tree in depth-first pre-order and emitting $1$ at each internal
node and $0$ at each leaf, together with a header recording $\tau$, the domain
half-width, $\varepsilon$, and the resolution parameters of \cref{sec:verify}.
\end{definition}

A reader given the header and the bits regenerates every box: the root is
determined by the header, and each $1$ instructs the reader to apply
\cref{def:split} and descend. No box is stored, so the record costs one bit per
node.

In practice the subdivision is run as $\NSeed$ independent subtrees, obtained by
splitting the root $\SeedDepth$ times before any bound is evaluated, and their
bit blocks are concatenated in order. This changes nothing about the tree and
allows the work to be divided. It gives a forest rather than a single tree, which
supplies a consistency check: a binary tree in which every internal node has two
children has one more leaf than it has internal nodes, so a forest of $s$ such
trees with $L$ leaves in total has exactly $2L-s$ nodes.

For $\tau=\MainBound$ the emitting run produced \Nodes{} nodes and \Leaves{}
leaves in \EmitMinutes{} minutes on eight cores, satisfying the identity above,
and no box reached the refinement guard of \Hmin{} at which the run would have
been reported as incomplete. The tree itself is \Nodes{} bits, or \TreeMB{}
megabytes. The file adds an index of the \NSeed{} block lengths, a further
\BlockIndexMB{} megabytes, and pads each block to a byte boundary, giving
\CertMB{} megabytes on disk and \BitsPerNode{} bits per node.

A second run of the same search, reported in \cref{tab:scaling}, prunes at
$\tau+10^{-9}$ instead and writes no certificate. It closes on \HardBoxes{}
boxes in \HardMinutes{} minutes. The two trees differ because the pruning rules
differ; the certificate on record is the first.

\subsection{Verification}
\label{sec:verify}

Checking the certificate means recomputing, at every leaf, a quantity that is
provably at most $\beta(\mathcal{B})$ and confirming that it still reaches $\tau$.
The disc is replaced by an inscribed polygon, each eroded
set by finitely many points of it, and the area of the resulting convex hull is
evaluated by the shoelace formula. The hull itself is computed by a library
routine on the production path and by a hand-written monotone chain on a sample
of leaves; which routine orders the vertices cannot affect
\cref{lem:verifierbound}, since dropping a vertex only shrinks the hull, which
errs in the direction the bound needs. Every point used is tested against the
definition of the set it is supposed to lie in, so the collection is an inner
approximation whatever the geometry does.

\begin{definition}
\label{def:witness}
A \emph{witness set} for a convex set $C$ is a finite subset of $C$. Given a box,
the verifier forms the witness set consisting of the vertices of the regular
$\WitnessM$-gon inscribed in $D$, together with, for each of $C_3$ and $C_5$, at
most $\WitnessK$ points constructed on the boundary of that set and retained only
if they satisfy $|p-V^{(n)}_j|\le1-\delta_n$ for every $j$.
\end{definition}

\begin{lemma}
\label{lem:verifierbound}
Let $W$ be a witness set formed as in \cref{def:witness} and let
$\alpha(W)=\operatorname{area}(\operatorname{conv}W)$. Then
$\alpha(W)\le\beta(\mathcal{B})$.
\end{lemma}

\begin{proof}
The vertices of an inscribed polygon lie in $D$. A retained point $p$ satisfies
$|p-V^{(n)}_j|\le1-\delta_n$ for every $j$, which is exactly the condition
$p\in C_n$ of \cref{prop:boxbound}. So $W\subseteq D\cup C_3\cup C_5$, hence
$\operatorname{conv}W\subseteq\operatorname{conv}(D\cup C_3\cup C_5)$, and areas
are monotone.
\end{proof}

An inner polygon spanned by points certified to lie in the test sets is the
device Xie~\cite{Xie2026} uses to bound the hull area from below on a parameter
domain, with the points carried as functions of the placement by outward-rounded
interval arithmetic and their cyclic order certified separately across the domain.
\Cref{lem:erosion} reaches the same end by other means. It replaces each moving
body by a fixed subset common to every placement in the box, so the witness points
do not vary with the parameters and no ordering has to be maintained across the
box. That is what admits test sets bounded by circular arcs rather than by
finitely many vertices.

Combining \cref{lem:verifierbound} with \cref{prop:boxbound}, a leaf at which
$\alpha(W)\ge\tau$ satisfies $\beta(\mathcal{B})\ge\tau$, which is what
\cref{prop:cover} requires. The retention test is what carries this argument, and
it cannot be replaced by an appeal to how the boundary points were constructed:
the construction traverses an arc whose angular span exceeds $\pi$ when
$1-\delta_n$ equals the circumradius of the corners, and points produced there
lie outside $C_n$. Testing each point removes any dependence on that geometry.

\begin{remark}
\label{rem:thresholds}
Two values must be kept apart. The subdivision was run against a threshold
slightly above $\tau$, chosen to anticipate the loss in \cref{lem:verifierbound}
and so to make the later check succeed. That threshold is a heuristic and carries
no part of the proof. What carries the proof is the verification, which tests
every leaf against $\tau$ itself. On this certificate the threshold was in fact
chosen too small to guarantee the outcome in advance, and the verifier records as
much before proceeding; every one of the \Leaves{} leaves cleared $\tau$
regardless, the smallest margin being \WorstSlack{}.
\end{remark}

\subsection{Arithmetic}
\label{sec:errors}

The quantity checked at a leaf is the area of the convex hull of finitely many
tested points, which limits how floating point can affect it. Two effects are
possible and both are bounded outright, so no appeal to the size of the observed
margin is needed.

A point is retained when its computed distance satisfies
$|p-V^{(n)}_j|\le1-\delta_n+\vartheta$ with $\vartheta=\Tau$, and a distance
computed as a square root of a sum of two squares carries relative error at most
$4\mathbf{u}$ with $\mathbf{u}=\Unit$. The corners and arc points entering that
distance are themselves computed from the sine and cosine of the box-centre
rotation, and library trigonometry is not correctly rounded in general, so a
further $8\mathbf{u}R$ is allowed for their coordinates, with
$R=\TmaxNaive+R_3$. A retained point can therefore lie outside $C_n$ by at most
$\varepsilon_0=\EpsMem$, the trigonometric term sitting three orders below
$\vartheta$ and not moving $\varepsilon_0$ at this precision. A set within $\varepsilon_0$ of a convex
region has hull area exceeding that of the region by at most
$\varepsilon_0 P+\pi\varepsilon_0^{2}$, where $P$ is a bound on the perimeter,
and the hull lies within distance $\TmaxNaive+R_3$ of the origin so
$P\le\PerimBound$. This contributes at most \ErrMembership{}. The shoelace sum
over $N$ vertices contributes at most $\gamma_N S$ with
$\gamma_N=N\mathbf{u}/(1-N\mathbf{u})$ and $S$ the sum of the magnitudes of its
terms. Both are bounded structurally rather than sampled. The witness set of
\cref{def:witness} holds at most $m+2K=\VertexBound$ points, so
$N\le\VertexBound$; every coordinate satisfies $|x|,|y|\le R$, so each of the $N$
terms is at most $R^{2}$ and $S\le NR^{2}=\ShoelaceS$. The contribution is at
most \ErrShoelace{}.

Writing $\eta=\Eta$ for the sum, every verified inequality holds with $\alpha(W)$
replaced by $\alpha(W)-\eta$. The smallest margin attained was \WorstSlack{},
larger than $\eta$ by a factor of \SlackOverEta{}, so \cref{prop:cover} applies
unchanged. Recomputing the leaf bound in \HighPrecDigits{} digit arithmetic on the
smallest boxes of the certificate reproduces the double precision value to
\HighPrecDiff{}, consistent with a bound of $\eta$.

\subsection{What the verification establishes}

The verifier reads the certificate and the header and derives everything else,
sharing no code with the subdivision, so \cref{prop:cover} is established from the
certificate alone. The certificates, the verifiers, the audits and every log cited
here are available at~\cite{Repo}, frozen at commit
\texttt{\RepoCommit}, which is the state these bounds were produced from.

Several further checks bear on whether the implementation computes what the
preceding sections describe. The bound of \cref{lem:verifierbound} is implemented
a second time against an independent convex hull routine, and \CrossCheckLeaves{}
leaves of this certificate were recomputed that way, the largest disagreement
being \CrossCheckWorst{}. The geometry kernel reproduces P\'al's $\PalLB$ to
\PalAgreement{} and the hull drawn in Figure~2 of Brass and
Sharifi~\cite{BrassSharifi2005} to \BSFigAgreement{}, neither constant having
entered its construction. It passes \KernelGates{} property tests, and each family
passes \AuditChecks{} further checks on \cref{lem:erosion} and \cref{def:split} at
the box sizes the certificate reaches. The forest identity of
\cref{def:certificate} holds exactly, so no subtree is missing from the record.

A last group tests the verifier itself, in both directions. Certificates altered
in their tree bits, truncated, or declaring a larger $\tau$, a smaller domain, or
a non-positive $\varepsilon$ are refused. Valid certificates are accepted, and so
are certificates declaring a finer witness resolution, which by
\cref{lem:verifierbound} yields a tighter bound rather than an invalid one. All
\ControlsPassed{} of these behave as required. A program that accepted every input
would pass the second half of that test and fail the first; one that refused every
input would do the reverse.

\begin{proof}[Proof of \cref{thm:main}]
Take $\tau=\MainBound$. The boxes at which the subdivision stopped cover
$\mathcal{D}$, by \cref{cor:leavescover}. At each of them the verification
evaluates $\alpha(W)$ and finds $\alpha(W)\ge\tau$, so $\beta\ge\tau$ there by
\cref{lem:verifierbound}, and the arithmetic of \cref{sec:errors} leaves this
conclusion intact. Since $f(\TmaxNaive)>\fAtTmax>\tau$, \cref{prop:cover} applies
and gives $\Lambda\ge\tau$.
\end{proof}

\section{Scope of the method}
\label{sec:scope}

Two questions decide how far this approach reaches. How much of the remaining gap
can be taken from the present test sets, and what happens if a fourth body is
added. Both are settled by measurement rather than by estimate, and the answers
are different in kind: the first is a matter of cost, the second is not.

\subsection{Cost against margin}

The subdivision must separate every placement from the value being certified, so
its cost is governed by the margin $\mu$ between $\tau$ and the smallest hull area
the family attains. \Cref{tab:scaling} takes that reference point to be a
numerically located optimum rather than the bound of \cref{prop:ceiling}, the two
differing by under $10^{-6}$ in the first family, which is far below the range of
$\mu$ tabulated. As $\tau$ approaches the ceiling the boxes near the
minimising placement must be refined further before their bound clears $\tau$,
and the count grows without bound. Running the same search at a range of targets
measures how fast.

\begin{table}[t]
\centering
\begin{tabular}{lrr@{\qquad}lrr}
\toprule
\multicolumn{3}{c}{$D,B_3,B_5$ \quad($d=5$)} &
\multicolumn{3}{c}{$D,B_3,B_5,B_7$ \quad($d=8$)}\\
\cmidrule(r){1-3}\cmidrule(l){4-6}
$\tau$ & $\mu$ & boxes & $\tau$ & $\mu$ & boxes\\
\midrule
0.8 & 0.034781 & 2{,}836 & 0.806495 & 0.030000 & 375{,}104\\
0.815 & 0.019781 & 20{,}496 & 0.812495 & 0.024000 & 913{,}876\\
0.822 & 0.012781 & 98{,}784 & 0.816495 & 0.020000 & 3{,}379{,}860\\
0.826 & 0.008781 & 345{,}964 & 0.820495 & 0.016000 & 11{,}151{,}480\\
0.829 & 0.005781 & 1{,}252{,}876 & 0.823495 & 0.013000 & 33{,}943{,}512\\
0.831 & 0.003781 & 4{,}075{,}580 & 0.825495 & 0.011000 & 82{,}262{,}500\\
0.8325 & 0.002281 & 14{,}078{,}604 &  & & \\
0.8344 & 0.000381 & 450{,}922{,}384 &  & & \\
\bottomrule
\end{tabular}
\caption{Boxes examined by the subdivision at a range of targets $\tau$. Every
row is the search of \cref{def:split} pruning at $\tau+10^{-9}$ and writing no
certificate, so the last row is the corroborating run of \cref{sec:computation}
rather than the emitting one. The margin $\mu$ is measured against a numerically
located optimum, \FiveDOptimum{} in the first family and \EightDOptimum{} in the
second, not against the proved bound of \cref{prop:ceiling}; in the first family
the two differ by under $10^{-6}$. Throughput was \ThroughputFive{} boxes per
second and \ThroughputEight{} respectively, on eight cores. Generated from
\texttt{scaling.log}.}
\label{tab:scaling}
\end{table}

\Cref{tab:scaling} records the counts and \cref{fig:scaling} plots them. The
eight-dimensional family is the steeper of the two by a wide margin, and it is
already the more expensive at every margin the two ranges share.

\begin{figure}[t]
\centering
\includegraphics{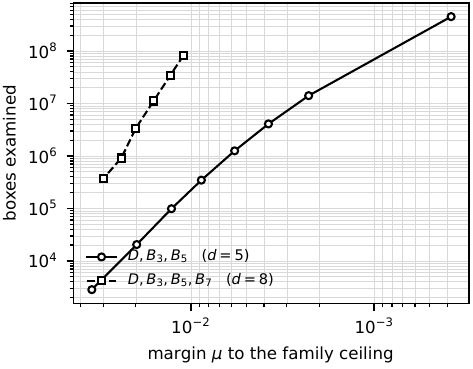}
\caption{Boxes examined against the margin $\mu$ to the family ceiling, for the
two families of \cref{tab:scaling}, on logarithmic axes with $\mu$ decreasing to
the right. The certified bound of \cref{thm:main} is the rightmost point of the
five-dimensional series.}
\label{fig:scaling}
\end{figure}

A power law $N\sim\mu^{-p}$ fits each
family, with $p=\FiveDExpGlobal$ over the five-dimensional range and
$p=\EightDExpGlobal$ over the eight-dimensional one. A rough expectation is
$p=d/2$, on the grounds that the bound fails only within a neighbourhood of the
minimising placement whose linear size scales as $\sqrt{\mu}$; that gives
\HalfDimFive{} and \HalfDimEight{} respectively. The five-dimensional fit sits
close to it and the eight-dimensional fit exceeds it.

The exponent is not constant. Between successive rows of the five-dimensional
column it rises from $3.50$ to $3.60$ and then falls to \FiveDExpTight{} between
the two tightest, so the neighbourhood in question is
not self-similar across scales and a single power law understates the cost at
wide margins and overstates it at narrow ones. Any extrapolation towards the
ceiling using a fit taken over the whole range is therefore conservative, which
is the direction that matters when the conclusion is that something is
unaffordable.

\subsection{What remains in the present family}

\Cref{prop:ceiling} applied to $D,B_3,B_5$ gives a ceiling of
\FamilyBCeilingSafe{}, so a margin of about $\HeadroomToCeiling$ is unclaimed.
Extrapolating from the last row of \cref{tab:scaling} at the local exponent
\FiveDExpTight{}, itself a two-point estimate from the two tightest margins,
certifying $\NextTarget$ would require roughly \NextTargetBoxes{} boxes, near
\NextTargetHours{} hours at the measured throughput. Nothing in the
argument changes; only the subdivision grows. We have not carried this out.

\begin{remark}
\label{rem:nosaving}
The sharper domain of \cref{rem:sharperdomain} does not reduce this cost. Run at
$\tau=\CheapTarget$ with every other parameter fixed, the subdivision examines
\BoxesCheapTarget{} boxes over the domain of \cref{cor:domain} and
\BoxesSharpDomain{} over the smaller one. The two a priori estimates agree
wherever the hull term in \cref{prop:boxbound} is the larger of the two, which is
the case throughout the neighbourhood of the minimising placement, and it is that
neighbourhood the box count is made of. A smaller domain removes boxes far from
the disc, which the subdivision was already discarding one at a time.
\end{remark}

\subsection{A fourth test body}

Adding the Reuleaux heptagon raises the attainable value to \EightDOptimum{},
this one located numerically rather than bounded as in \cref{prop:ceiling}, and
raises the parameter count from five to eight, which is where the method stops. Taking the
eight-dimensional fit of \cref{tab:scaling} at face value, certifying
$\tau=\MainBound$ in that family calls for about \EightDBoxes{} boxes, near
\EightDDays{} days at the measured throughput. Allowing for the flattening seen
in the five-dimensional column reduces that figure substantially without
approaching feasibility.

The obstruction is cost and not principle. Every statement in
\cref{sec:prelim,sec:ceiling,sec:reduction} holds verbatim for any finite family
of Reuleaux polygons, \cref{lem:erosion} included, and the subdivision would
terminate. What fails is that the number of boxes required grows faster than the
dimension is worth. Reaching the value that Gibbs~\cite{Gibbs2014} estimates for
a family of five bodies would need eleven parameters, further still.

\section{Related problems}
\label{sec:related}

The division of labour in the proof is worth stating plainly.
\Cref{sec:prelim,sec:ceiling,sec:reduction} are analytic and contain no
computation: the admissibility of the test sets rests on \cref{lem:width}, the
ceiling of \cref{prop:ceiling} on one exhibited arrangement, and the reduction on
\cref{lem:apriori,lem:erosion,lem:displacement}. The computation enters once, in
\cref{sec:computation}, as a subdivision of one compact box, and what it
establishes is the hypothesis of \cref{prop:cover} and nothing else.

Two of our constants are sharp for the method and no advance in certification
changes them. \Cref{prop:ceiling} bounds what the classical test sets can yield at
\ClassicalCeilingSafe{}, and the same argument bounds our own family at
\FamilyBCeilingSafe{}. These are properties of the test sets, not of the search,
so a faster machine or a better subdivision leaves them where they are. Passing
either requires enlarging the family, and the constant-width bodies are already
the maximal test sets by \cref{rem:completion}, so the only remaining freedom is
to use more of them or to use ones that are not regular.

Against that, the obstruction to going further is cost rather than principle, and
the two should not be confused. Within the present family a margin of about
$\HeadroomToCeiling$ is unclaimed, and \cref{sec:scope} prices it. For a fourth
test body every statement in \cref{sec:prelim,sec:ceiling,sec:reduction} holds
verbatim, the subdivision terminates, and only the box count stands in the way.
Brass and Sharifi~\cite{BrassSharifi2005} record the same barrier for their
polygonal family and leave it as the obstacle to further progress;
\cref{tab:scaling} measures where it now sits.

\Cref{lem:erosion} is not specific to Reuleaux polygons. Its proof uses nothing
about the points $V_j$ beyond their number being finite, and in fact not even
that: for an arbitrary index set the same triangle inequality gives
$\bigcap_{i}\overline{D}(V_i,1-\delta)\subseteq\bigcap_i\overline{D}(W_i,1)$
whenever $|W_i-V_i|\le\delta$. A body of constant width one is the intersection of
the closed unit discs centred at its own points~\cite{Grunbaum1963}, so the
estimate applies to every admissible test set of \cref{rem:completion}, not only
to the ones used here. Irregular Reuleaux polygons form an $(n-3)$-parameter
family at fixed $n$, and optimising over shape as well as position is a natural
next question that the present method can express without modification, at a cost
this paper does not attempt.

What would change the picture is a way to bound $M$ from below without exhausting
the placement space. Every lower bound for this problem since 1920, ours included,
controls all placements at once by subdividing them, and the cost of that is what
\cref{tab:scaling} records. A dual or variational argument,
certifying a minimum without enumerating a neighbourhood of it, would remove the
dimensional barrier rather than push it, and would reach the families whose
numerical values Gibbs~\cite{Gibbs2014} reports. The gap
$[\MainBound,\GibbsUpper]$ stands until one or the other side moves.

\bibliography{refs}

@article{Pal1920,
  author  = {P{\'a}l, Julius},
  title   = {{\"U}ber ein elementares {V}ariationsproblem},
  journal = {Danske Mat.-Fys. Meddelelser},
  volume  = {III},
  number  = {2},
  pages   = {1--35},
  year    = {1920}
}

@article{Sprague1936,
  author  = {Sprague, Roland},
  title   = {{\"U}ber ein elementares {V}ariationsproblem},
  journal = {Matematisk Tidsskrift B},
  pages   = {96--99},
  year    = {1936},
  note    = {JSTOR 24530328}
}

@article{Hansen1992,
  author  = {Hansen, Hans Christian},
  title   = {Small universal covers for sets of unit diameter},
  journal = {Geometriae Dedicata},
  volume  = {42},
  number  = {2},
  pages   = {205--213},
  year    = {1992},
  doi     = {10.1007/BF00147549}
}

@article{Duff1980,
  author  = {Duff, George F. D.},
  title   = {A smaller universal cover for sets of unit diameter},
  journal = {C. R. Math. Rep. Acad. Sci. Canada},
  volume  = {2},
  pages   = {37--42},
  year    = {1980}
}

@article{Elekes1994,
  author  = {Elekes, Gy{\"o}rgy},
  title   = {Generalized breadths, circular {C}antor sets, and the least area {UCC}},
  journal = {Discrete \& Computational Geometry},
  volume  = {12},
  pages   = {439--449},
  year    = {1994},
  doi     = {10.1007/BF02574391}
}

@article{BrassSharifi2005,
  author  = {Brass, Peter and Sharifi, Mehrbod},
  title   = {A lower bound for {L}ebesgue's universal cover problem},
  journal = {International Journal of Computational Geometry \& Applications},
  volume  = {15},
  number  = {5},
  pages   = {537--544},
  year    = {2005},
  doi     = {10.1142/S0218195905001828}
}

@article{BaezBagdasaryanGibbs2015,
  author  = {Baez, John C. and Bagdasaryan, Karine and Gibbs, Philip},
  title   = {The {L}ebesgue Universal Covering Problem},
  journal = {Journal of Computational Geometry},
  volume  = {6},
  number  = {1},
  pages   = {288--299},
  year    = {2015},
  eprint  = {1502.01251},
  archivePrefix = {arXiv}
}

@misc{Gibbs2014,
  author = {Gibbs, Philip},
  title  = {A New Slant on {L}ebesgue's Universal Covering Problem},
  year   = {2014},
  eprint = {1401.8217},
  archivePrefix = {arXiv}
}

@misc{Gibbs2018,
  author = {Gibbs, Philip},
  title  = {An Upper Bound for {L}ebesgue's Covering Problem},
  year   = {2018},
  eprint = {1810.10089},
  archivePrefix = {arXiv}
}

@incollection{Grunbaum1963,
  author    = {Gr{\"u}nbaum, Branko},
  title     = {Borsuk's problem and related questions},
  booktitle = {Convexity},
  series    = {Proceedings of Symposia in Pure Mathematics},
  volume    = {VII},
  pages     = {271--284},
  publisher = {American Mathematical Society},
  year      = {1963}
}

@inproceedings{FeketeEtAl2021,
  author    = {Fekete, S{\'a}ndor P. and Gurunathan, Vijaykrishna and
               Juneja, Kushagra and Keldenich, Phillip and Kleist, Linda and
               Scheffer, Christian},
  title     = {Packing Squares into a Disk with Optimal Worst-Case Density},
  booktitle = {37th International Symposium on Computational Geometry (SoCG 2021)},
  series    = {LIPIcs},
  volume    = {189},
  pages     = {36:1--36:16},
  year      = {2021},
  doi       = {10.4230/LIPIcs.SoCG.2021.36}
}

@article{HalesEtAl2017,
  author  = {Hales, Thomas and Adams, Mark and Bauer, Gertrud and
             Dang, Tat Dat and Harrison, John and Hoang, Le Truong and
             Kaliszyk, Cezary and Magron, Victor and McLaughlin, Sean and
             Nguyen, Tat Thang and Nguyen, Quang Truong and Nipkow, Tobias and
             Obua, Steven and Pleso, Joseph and Rute, Jason and
             Solovyev, Alexey and Ta, Thi Hoai An and Tran, Nam Trung and
             Trieu, Thi Diep and Urban, Josef and Vu, Ky and
             Zumkeller, Roland},
  title   = {A formal proof of the {K}epler conjecture},
  journal = {Forum of Mathematics, Pi},
  volume  = {5},
  year    = {2017}
}

@misc{Xie2026,
  author = {Xie, Niantao},
  title  = {A Certified Lower Bound for {L}ebesgue's Universal Cover Problem},
  year   = {2026},
  eprint = {2606.04458},
  archivePrefix = {arXiv},
  primaryClass  = {cs.CG},
  note   = {preprint, version 4, 8 July 2026}
}

@book{KellyWeiss1979,
  author    = {Kelly, Paul J. and Weiss, Max L.},
  title     = {Geometry and Convexity: A Study in Mathematical Methods},
  publisher = {Wiley},
  year      = {1979}
}

@article{Hansen1975,
  author  = {Hansen, Hans Christian},
  title   = {A small universal cover of figures of unit diameter},
  journal = {Geometriae Dedicata},
  volume  = {4},
  pages   = {165--172},
  year    = {1975}
}

@book{BrassMoserPach2005,
  author    = {Brass, Peter and Moser, William O. J. and Pach, J{\'a}nos},
  title     = {Research Problems in Discrete Geometry},
  publisher = {Springer},
  year      = {2005}
}

@article{AppelHaken1977a,
  author  = {Appel, Kenneth and Haken, Wolfgang},
  title   = {Every planar map is four colorable. {P}art {I}. {D}ischarging},
  journal = {Illinois Journal of Mathematics},
  volume  = {21},
  pages   = {429--490},
  year    = {1977}
}

@article{AppelHaken1977b,
  author  = {Appel, Kenneth and Haken, Wolfgang},
  title   = {Every planar map is four colorable. {P}art {II}. {R}educibility},
  journal = {Illinois Journal of Mathematics},
  volume  = {21},
  pages   = {491--567},
  year    = {1977}
}

@article{RobertsonEtAl1997,
  author  = {Robertson, Neil and Sanders, Daniel and Seymour, Paul and Thomas, Robin},
  title   = {The four-colour theorem},
  journal = {Journal of Combinatorial Theory, Series B},
  volume  = {70},
  pages   = {2--44},
  year    = {1997}
}

@article{SzekeresPeters2006,
  author  = {Szekeres, George and Peters, Lindsay},
  title   = {Computer solution to the 17-point {E}rd{\H o}s--{S}zekeres problem},
  journal = {The ANZIAM Journal},
  volume  = {48},
  number  = {2},
  pages   = {151--164},
  year    = {2006}
}

@article{HassSchlafly2000,
  author  = {Hass, Joel and Schlafly, Roger},
  title   = {Double bubbles minimize},
  journal = {Annals of Mathematics},
  volume  = {151},
  number  = {2},
  pages   = {459--515},
  year    = {2000}
}

@book{MartiniMontejanoOliveros2019,
  author    = {Martini, Horst and Montejano, Luis and Oliveros, D{\'e}borah},
  title     = {Bodies of Constant Width: An Introduction to Convex Geometry
               with Applications},
  publisher = {Birkh{\"a}user},
  year      = {2019},
  doi       = {10.1007/978-3-030-03868-7},
  note      = {MR3930585, Zbl 1468.52001}
}

@misc{Repo,
  author       = {Mishra, Ujjwal},
  title        = {Certificates and verifiers for a lower bound on {L}ebesgue's
                  universal covering problem},
  year         = {2026},
  howpublished = {\url{https://github.com/Ujjwal238/universal-cover-problem}},
  note         = {commit da956b117a5f876fc3c530db3e00673e6cb18865}
}
\end{document}